\documentclass{lmcs}
\pdfoutput=1

\usepackage{lastpage}
\lmcsdoi{16}{1}{17}
\lmcsheading{}{\pageref{LastPage}}{}{}%
{Mar.~21,~2019}{Feb.~14,~2020}{}

\usepackage{graphicx}
\usepackage{amssymb}
\usepackage[mathscr]{euscript}
\usepackage{amsmath}
\usepackage{amsfonts}
\usepackage{mathtools}
\usepackage{amsthm}
\usepackage{indentfirst}
\usepackage{tikz}
\usetikzlibrary{arrows,automata,positioning}

\theoremstyle{plain}
\newtheorem*{thmA}{Theorem A}
\newtheorem*{thmB}{Theorem B}

\newcommand{\N}{\mathbb{N}}
\newcommand{\R}{\mathbb{R}}

\newcommand{\Q}{\mathbb{Q}}
\newcommand{\calA}{\mathcal{A}}
\newcommand{\calG}{\mathcal{G}}
\newcommand{\calH}{\mathcal{H}}
\newcommand{\Diam}{\operatorname{Diam}}
\begin{document}

\author[A.~Block Gorman]{Alexi Block Gorman}
\address{Department of Mathematics, University of Illinois at Urbana-Champaign, Urbana}
\email{\{atb2,phierony,eakapla2,ruoyum2,erikw,zwang199,zxiong8,hyang87\}@illinois.edu}
\urladdr{http://faculty.math.illinois.edu/\textasciitilde{\normalfont\{}phierony,eakapla2,erikw{\normalfont\}}\texorpdfstring{\vspace*{-15mm}}{}}

\author[P.~Hieronymi]{Philipp Hieronymi}
\author[E.~Kaplan]{Elliot Kaplan}
\author[R.~Meng]{Ruoyu Meng}
\author[E.~Walsberg]{Erik Walsberg}
\author[Z.~Wang]{Zihe Wang}
\author[Z.~Xiong]{Ziqin Xiong}
\author[H.~Yang]{Hongru Yang}

\renewcommand{\shortauthors}{A.~Block Gorman et al.}

\title{Continuous Regular Functions}

\begin{abstract}
Following Chaudhuri, Sankaranarayanan, and Vardi, we say that a function $f:[0,1] \to [0,1]$ is $r$-regular if there is a B\"{u}chi automaton that accepts precisely the set of base $r \in \mathbb{N}$ representations of elements of the graph of $f$.
We show that a continuous $r$-regular function $f$ is locally affine away from a nowhere dense, Lebesgue null, subset of $[0,1]$. As a corollary we establish that every differentiable $r$-regular function is affine. It follows that checking whether an $r$-regular function is differentiable is in $\operatorname{PSPACE}$. Our proofs rely crucially on connections between automata theory and metric geometry developed by Charlier, Leroy, and Rigo.
\end{abstract}

\maketitle

\section{Introduction}\label{s:intro}
The study of regular real analysis was proposed by  Chaudhuri, Sankaranarayanan, and Vardi in~\cite{CSV} as the analysis of real functions whose graphs are encoded by automata on infinite words. The motivation for such an investigation is to produce automata-theoretic decision procedures deciding continuity and other properties of regular functions. As part of that study in~\cite{CSV}, the authors  establish novel topological-geometric results about regular functions, in particular Lipschitzness and H\"older-continuity. This investigation is part of a larger enterprise of studying subsets of $\R^n$ described in terms of such automata, and has seen various applications in the verification of systems with unbounded mixed variables taking integer or real values (see Boigelot et al.~\cite{BBR,BJW,BRW}). Recently, there has been a new focus on the study of topological and metric-geometrical properties of such sets (see Adamczewski and Bell~\cite{AdamB} and Charlier, Leroy, and Rigo~\cite{CLR}). We continue the investigation by  giving answers to two questions raised in~\cite{CSV} in the case of regular functions with equal input and output base:
\begin{enumerate}
\item Is there a simple characterization of continuous or differentiable regular functions?
\item Is it decidable to check whether a regular function is differentiable?
\end{enumerate}

\noindent Let $r,s \geq 1$ be positive integers, set $\Sigma_r=\{0,\dots, r-1\}$ and $\Sigma_s = \{0,\dots,s-1\}$. A function $f: [0,1]\to [0,1]$ is \textbf{$(r,s)$-regular} if there is a non-deterministic B\"uchi automaton in the language $\Sigma_r \times \Sigma_s$ accepting precisely all words $(w_{1,1},w_{2,1})(w_{1,2},w_{2,2})\dots \in {(\Sigma_r \times \Sigma_s)}^{\omega}$ for which there are $a,b \in [0,1]$ such that $0.w_{1,1}w_{1,2}\dots$ is a base $r$ expansion of $a$, $0.w_{2,1}w_{2,2}\dots$ is a base $s$ expansion of $b$, and $f(a)=b$. If $r=s$, we simply say $f$ is $r$-regular.

\medskip
An instructional example of a continuous $3$-regular function is the function that maps $x\in [0,1]$ to the distance between $x$ and the classical middle-thirds Cantor set; see Figure~\ref{fig:modcantor}. Observe that this function is locally affine away from a nowhere dense set. The main technical result of this paper shows that this holds for all continuous $r$-regular functions. We say that a function $g : I \to \mathbb{R}$ on an interval $I$ is $\mathbb{Q}$-\textbf{affine} if it is of the form $g(t) = \alpha t + \beta$ for some $\alpha, \beta \in \mathbb{Q}$.

\begin{thmA}\label{thm:main}
Let $f : [0,1] \to [0,1]$ be a continuous, $r$-regular function. Then there is a nowhere dense Lebesgue null set $C \subseteq [0,1]$ such that $f$ is locally $\mathbb{Q}$-affine away from $C$.
\end{thmA}

\noindent We know that Theorem A fails for $(r,s)$-regular functions when $r> s$.

\begin{thmB}
A differentiable $r$-regular function is $\Q$-affine. Given a B\"uchi automaton that recognizes an $r$-regular function $f:[0,1]\to [0,1]$, the problem of checking whether $f$ is differentiable, is in $\operatorname{PSPACE}$.
\end{thmB}

\noindent We do not know whether Theorem B also fails in general, nor whether Question (2) has a positive answer in general. Furthermore, Theorem A gives an affirmative answer to a special case of an important question in mathematical logic about continuous definable functions in first-order expansions of the ordered additive group of real numbers. We discuss this question in Section~\ref{section:application}.

\noindent Theorem B improves several results in the literature. Let $f: [0,1]\to [0,1]$ be $r$-recognizable. In~\cite{HW-continuous} a model-theoretic argument was used to show that $f$ is affine whenever $f$ is continuously differentiable. Similar results had been proven before for $C^2$ functions, or for more restricted classes of automata, by Anashin~\cite{Anashin-auto}, Kone\v{c}n\'y~\cite{Konecny-auto}, and Muller~\cite{Muller-auto}. In unpublished work, Bhaskar, Moss, and Sprunger~\cite{BMS} also showed that any differentiable function computable by a letter-to-letter transducer is affine.

\medskip
Our proof of Theorem A depends on three ingredients. The first is the theorem, proven in~\cite{CSV}, that every continuous $r$-regular function $[0,1] \to \mathbb{R}$ is Lipschitz.
The second ingredient is the main result from~\cite{CLR} which implies that the graph of an $r$-regular function is the attractor of a graph-directed iterative function system. This observation allows us to use the full power of metric geometry to study regular functions. In particular, we adjust to our setting the proof of the famous result of Hutchinson~\cite[Remark 3.4]{Hutch} that a Lipschitz curve in $\R^2$ which is the attractor of a linear iterated function system is a line segment.

\begin{figure}[t]\label{fig:modcantor}
\includegraphics[width=0.5\textwidth]{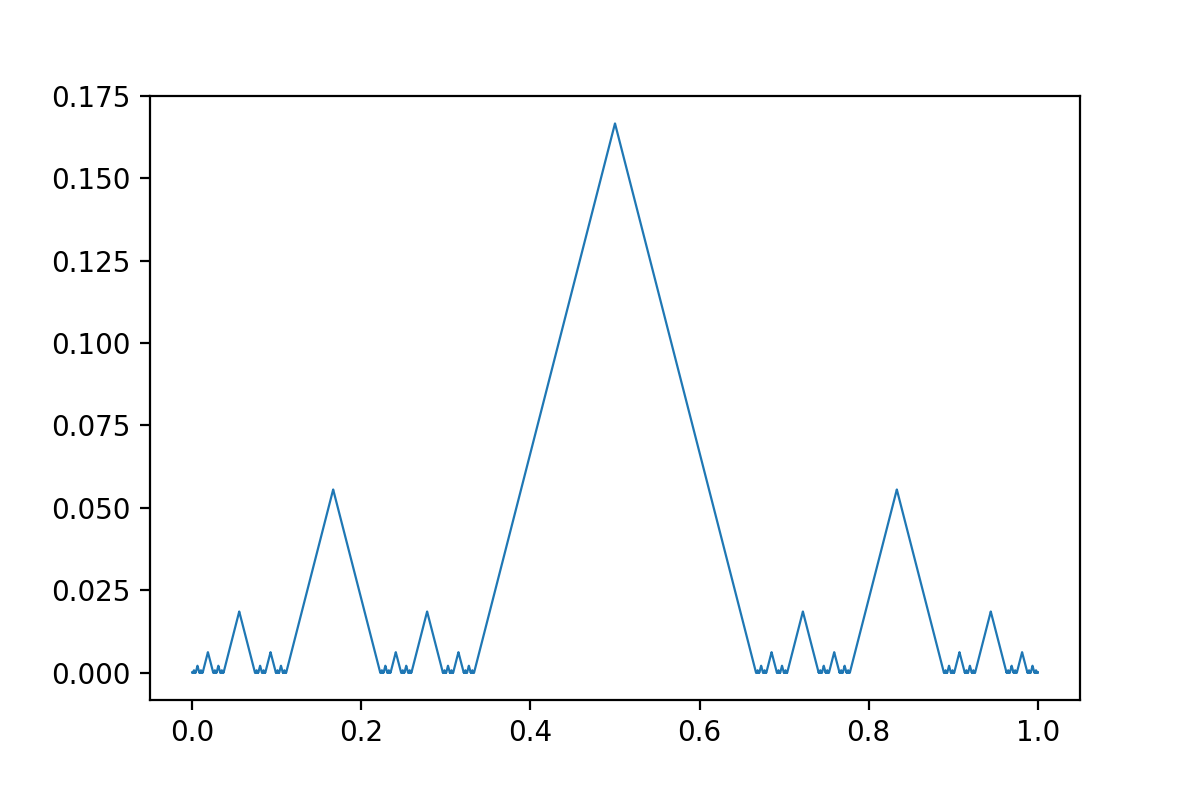}
\caption{The distance from the middle-thirds Cantor set function is $3$-regular.}
\end{figure}

\subsection*{Acknowledgments} This work was done in the research project ``Automata and Differentiable Functions'' at the Illinois Geometry Lab in Spring 2018. R.M., Z.W., Z.X. and H.Y. participated as undergraduate scholars, A.B.G. and E.K. served as graduate student team leaders, and P.H. and E.W. as faculty mentors. A.B.G. was supported by an NSF Graduate Research Fellowship. P.H. was partially supported by NSF grant DMS-1654725.

\section{Preliminaries}

Throughout, $i,j,k,\ell,m,n$ are used for natural numbers and $\delta,\varepsilon$ are used for positive real numbers. We denote the coordinate projection from $\R^2$ to $\R$ onto the first coordinate by $\pi$. Given a function $f$, we denote its graph by $\Gamma(f)$.  Given a (possibly infinite word) $w$ over an alphabet $\Sigma$, we write $w_i$ for the $i$-th letter of $w$, and $w|_n$ for $w_1\dots w_n$. We denote the set of infinite words over $\Sigma$ by $\Sigma^{\omega}$.

\subsection{B\"uchi Automata}
A \textbf{B\"uchi automaton (over an alphabet $\Sigma$)} is a quintuple $\calA = (Q,\Sigma,\Delta,I,F)$ where $Q$ is a finite set of states, $\Sigma$ is a finite alphabet, $\Delta\subseteq Q \times \Sigma \times Q$ is a transition relation, $I\subseteq Q$ is a set of initial states, and $F \subseteq Q$ is a set of accept states.
We fix a B\"uchi automaton $\calA = (Q,\Sigma,\Delta,I,F)$ for the remainder of this section.
We say that $\calA$ is \textbf{deterministic} if $I$ is a singleton and for all $q\in Q$ and all $\sigma\in \Sigma$, there is at most one $p \in Q$ such that $(q,\sigma,p) \in \Delta$.

\medskip
Throughout a \textbf{graph} is a directed multigraph; that is a set of vertices $V$ and a set of edges $E$, together with source and target maps $s,t : E \to V$. Vertices $u,v$ are connected by an edge $e$ if $s(e) = u$ and $t(e) = v$.
The graph underlying $\calA$ is $(Q, \Delta, s, t)$ where $s(p,\sigma,q) = p$ and $t(p,\sigma, q) = q$ for all $(p, \sigma, q) \in \Delta$.

\medskip
Let $p$ and $q$ be states in $Q$. A \textbf{path} of length $n$ from $p$ to $q$ is a sequence
\[
\left((p,\sigma_1,s_1),(s_1,\sigma_2,s_2),\ldots,(s_{n-1},\sigma_{n},q)\right) \in \Delta^{<\omega}.
\]
We call the word $\sigma_1\dots \sigma_{n}$ the label of the path. We denote by $P_{p,q}(n)\subseteq \Sigma^n$ the set of all labels of paths of length $n$ from $p$ to $q$. Note that  $P_{p,q}(0)$ is the set containing just the empty string if $p= q$, and $P_{p,q}$ is the empty set otherwise. 

\medskip
Let $\sigma$ be a finite or infinte word over alphabet $\Sigma$. A \textbf{run of $\sigma$ from $p$} is a finite or infinite sequence $s$ of states in $Q$ such that $s_0 = p$, $(s_n,\sigma_n,s_{n+1})\in \Delta$ for all $n<|\sigma|$
and, $|s|=|\sigma|+1$ if $|\sigma|$ is finite, and infinite otherwise. If $p \in I$, we say $s$ is a \textbf{run of $\sigma$}. Then $\sigma$ is \textbf{accepted by $\calA$} if there is a run $s_0s_1\ldots$ of $\sigma$ such that $\{n:s_n \in F\}$ is infinite.
We let $L(\calA)$ be the set of words accepted by $\calA$. 

\medskip Let $q \in Q$ be a state of $\calA$. We say that $q$ is \textbf{accessible} if there is a path from $I$ to $q$, and that $q$ is \textbf{co-accessible} if there is a path from $q$ to an accept state.
We say that $\calA$ is \textbf{trim} if all states are accessible and co-accessible. By removing all states that are not both accessible and co-accessible from $\calA$, we obtain a trim B\"uchi Automaton $\calA'$ that accepts exactly the same words as $\calA$.

\medskip
We say that $\calA$ is \textbf{closed} if
\[
\{q \in Q:\text{there is a nonempty path from }q\text{ to }q\}\subseteq F.
\]
The \textbf{closure} of $\calA$ is the B\"uchi automaton $\bar{\calA} = (Q,\Sigma,\Delta,I,Q)$; that is the automaton in which all states are changed to accept states. Note that the closure of $\calA$ is closed.

\medskip
Two elements of $p,q$ of $Q$ \textbf{lie in the same strongly connected component} if there is a path from $p$ to $q$ and vice versa. Being in the same strongly connected component is an equivalence relation on $Q$. We call each equivalence class a \textbf{strongly connected component} of $\calA$. We say that $\calA$ is \textbf{strongly connected} if it only has one strongly connected component.
The \textbf{condensation} of $\calA$ is the graph $H$ whose vertices are the strongly connected components of the underlying graph of $\calA$ such that two strongly connected components $C_1,C_2$ are connected by an edge if
\begin{itemize}
\item $C_1 \neq C_2$,
\item there are vertices $q_1 \in C_1, q_2 \in C_2$ and a path from $q_1$ to $q_2$.
\end{itemize}
A \textbf{sink} of $\calA$ is a strongly connected compontent of $\calA$ that has no outgoing edges in the condensation of $\calA$.

\medskip
We say that $\calA$ is \textbf{weak} if
 whenever $p,q \in Q$ belong to the same strongly connected component,
$p\in F$ if and only if $q\in F$. Observe that every closed automaton is weak. We will use the fact that whenever $\calA$ is weak, there is a weak deterministic automata $\calA'$ such that $L(\calA)=L(\calA')$ (see Boigelot, Jodogne, and Wolper~\cite[Section 5.1]{BJW}).

\subsection{Recognizable subsets of \texorpdfstring{${[0,1]}^n$}{[0,1]\textasciicircum{}n}}
Given an infinite word $w$ over $\Sigma_r$, we let
\[
v_r(w) := \sum_{i=1}^{\infty} \frac{w_{i-1}}{r^{i}} \in [0,1].
\]
Let $n \geq 1$ and $\vec{r}=(r_1,\dots,r_n)\in \N_{>1}^n$. Set $\Sigma_{\vec{r}}:=\Sigma_{r_1}\times \cdots \times \Sigma_{r_n}$.

\medskip
Given $(w_1,\ldots,w_n) \in \Sigma_{r_1}^{\omega}\times \cdots \times \Sigma_{r_n}^{\omega}$, we denote by $w_1\times \cdots \times w_n$ the word $u=u_1u_2\dots \in \Sigma_{\vec{r}}^{\omega}$ such that for each $i\in \N$
\[
u_i = (w_{1,i},\dots,w_{n,i}) \in \Sigma_{\vec{r}}.
\]
If $u=w_1\times \cdots \times w_n$, we set
\[
v_{\vec{r}}(u) := (v_{r_1}(w_1),\ldots,v_{r_n}(w_n)) \in {[0,1]}^n.
\]
A set $X \subseteq {[0,1]}^n$ is \textbf{$\vec{r}$-recognized} by a  B\"uchi automaton $\calA$ over $\Sigma_{\vec{r}}$ if
\[
L(\calA) = \{ u \in \Sigma_{\vec{r}}^{\omega} \ : \ v_{\vec{r}}(u) \in X\}.
\]
We say $X$ is \textbf{$\vec{r}$-recognizable} if $X$ is $\vec{r}$-recognized by some B\"uchi automaton over $\Sigma_{\vec{r}}$.
If $r = r_1 = \cdots = r_{n}$, then we use $r$-recognized and $r$-recognizable.

\medskip
Recognizable sets are introduced in~\cite{BBR} and their connection to first-order logic is studied in~\cite{BRW}.  By~\cite[Theorem 6]{BBR} $r$-recognizable sets are closed under boolean combinations and coordinate projections. In this paper we only consider bounded recognizable sets. However, the notion of recognizable sets can be extended to unbounded subsets of $\R^n$ as is done in~\cite{BBR} and~\cite{CSV}\footnote{Recognizable sets are called regular sets in~\cite{CSV}.}. The way how recognizability is defined on unbounded sets slightly differs between~\cite{BBR} and~\cite{CSV}, but fortunately the two definitions coincide for bounded sets.

\medskip Let $X$ be a closed subset of ${[0,1]}^n$.
The collection
\[
\left\{ r^k \left( \prod_{ i = 1 }^{ n } \left[ \frac{ m_i }{ r^k } , \frac{ m_i + 1 }{ r^k }\right] \cap X \right) - (m_1 , \ldots , m_n)   : k \in \mathbb{N}, 0 \leq m_1,\ldots, m_n \leq r^k \right\}
\]
of subsets of ${[0,1]}^n$ is known as the \textbf{$r$-kernel} of $X$.
We say that $X$ is \textbf{$r$-self similar} if the $r$-kernel of $X$ contains only finitely many distinct sets. This notion was introduced in~\cite{AdamB}.
Fact~\ref{fact:similar} gives a purely geometric characterization of closed $r$-recognizable sets.
While not crucial to the main argument of the paper, we will use this fact to show that the nowhere dense $C$ in the statement of Theorem A is Lebesgue null.

\begin{factC}[{\cite[Theorems 57 \& 62]{CLR}}]\label{fact:similar} Let $X$ be a closed subset of ${[0,1]}^n$. Then $X$ is $r$-recognizable if and only if $X$ is $r$-self similar.
\end{factC}

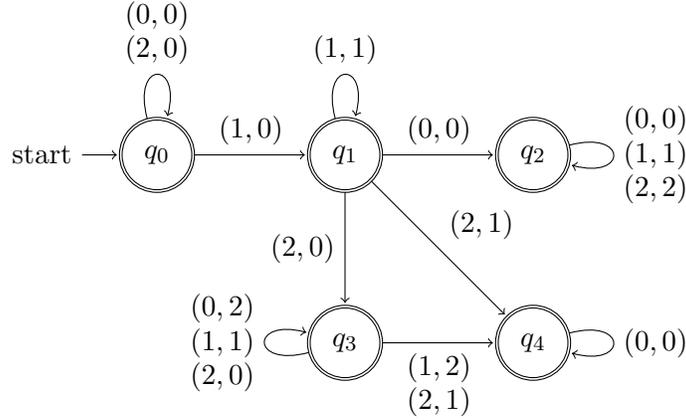
\begin{figure}[t]
\begin{tikzpicture}[node distance=2.5cm,on grid,auto, shorten > = 1pt]
	\node[state,initial,accepting](q0){$q_0$};
    \node[state,accepting](q1)[right=of q0]{$q_1$};
	\node[state,accepting](q2)[right=of q1]{$q_2$};
    \node[state,accepting](q3)[below=of q1]{$q_3$};
   \node[state,accepting](q4)[right=of q3]{$q_4$};
    \path[->]
    (q0) edge [loop above] node [align=center] {$(0,0)$\\$(2,0)$}()
   	edge node {$(1,0)$} (q1)
    (q1) edge [loop above] node [align=center] {$(1,1)$}()
    	edge node {$(0,0)$} (q2)
   	edge node {$(2,1)$} (q4)
    	edge node [swap] {$(2,0)$} (q3)
    (q2) edge [loop right] node [align=center] {$(0,0)$\\$(1,1)$\\$(2,2)$}()
   (q3) edge [loop left] node [align=center] {$(0,2)$\\$(1,1)$\\$(2,0)$}()
    	edge node [align=center,swap] {$(1,2)$\\$(2,1)$} (q4)
   (q4) edge [loop right] node [align=center] {$(0,0)$}();
\end{tikzpicture}
\caption{A B\"uchi automaton that $3$-recognizes the graph of $d$, the distance function from the middle-thirds Cantor set.}%
\label{figure:distancecantor}
\end{figure}

\noindent
We say that $X \subseteq {[0,1]}^n$ is \textbf{$\vec{r}$-accepted} by a B\"uchi automaton $\calA$ if
\[
X =  v_{\vec{r}}(L(\calA)), \text{ that is } X = \{ v_{\vec{r}}(u) \ : u \in \Sigma_{\vec{r}}^{\omega} \text{ is accepted by } \calA\}.
\]
Again, if $r=r_1 = \cdots = r_n$, we write $r$-accepted. It is clear that an $\vec{r}$-recognizable set is $\vec{r}$-acceptable. The converse is also true, but we will not apply this fact. We distinguish between $r$-recognizability and $r$-acceptance because we can prove stronger theorems about automata that $r$-accept a given set.

\medskip
Given a  B\"uchi automaton $\calA$ over $\Sigma_{\vec{r}}$, we set
\[
V_{\vec{r}}(\calA) := v_{\vec{r}}(L(\calA)).
\]
When $r=r_1=\cdots=r_n$, we will simply write $V_r(\calA)$.

\medskip
By~\cite[Lemma 58]{CLR} if $X\subseteq \R^n$ is $\vec{r}$-accepted by a closed and trim B\"uchi automaton, then $X$ is closed in the usual topology on $\R^n$.

\begin{factC}[{\cite[Remark 59]{CLR}}]\label{closed}
Let $X\subseteq \R^n$ be $\vec{r}$-accepted by a trim B\"uchi automaton $\calA$. Then the topological closure $\overline{X}$ of $X$ is $\vec{r}$-accepted by the closure $\bar{\calA}$ of $\calA$.
\end{factC}

\subsection{Regular functions} A function $f:{[0,1]}^m \to {[0,1]}^{n-m}$ is said to be \textbf{$\vec{r}$-regular} if its graph $\Gamma(f) \subseteq {[0,1]}^n$ is $\vec{r}$-recognizable. If $r=r_1=\cdots=r_n$, we simply write $r$-regular.

\medskip
It is well-known that every $\Q$-affine function $f: [0,1] \to [0,1]$ is $r$-regular for every $r\geq 2$ (see for example~\cite[Theorem 4 \& 5]{CSV} or~\cite[Theorem 5 \& 6]{BBR}). We now give an example of a non-affine continuous $3$-recognizable function.

\begin{exa} Let $C \subseteq [0,1]$ be the classical middle-thirds Cantor set. Consider the function $d:[0,1] \rightarrow [0,1]$ given by
\[
 d(x) = \min \{ |x - y| : y \in C \} \quad \text{for all} \quad 0 \leq x \leq 1;
\]
i.e. $d$ maps $x$ to the distance between $x$ and $C$ (see Figure~\ref{fig:modcantor} for a plot). It is not hard to see that $d$ is $3$-regular. An automaton $\calA$ that $3$-recognizes the graph is displayed in Figure~\ref{figure:distancecantor}.

\noindent The graph of $d$ is $3$-accepted by the slightly simpler automaton $\calA_2$ in Figure~\ref{figure:distancecantor2}. Observe that in this automaton the sets $\{q_2\}$ and $\{q_3\}$ are sinks of $\calA_2$. Furthermore, whenever there is $x\in [0,1]$ such that $d$ is locally affine on an interval around $x$ with slope $1$, then there is
$w \in L(\calA_2)$ such that
\begin{enumerate}
\item $v_{(3,3)}(w) = (x,d(x))$ and
\item the acceptance run of $w$ is eventually in $q_2$.
\end{enumerate}
Similarly, if $d$ is locally affine on an interval around $x$ with slope $-1$, the acceptance run of such a $w$ satisfying (1) is eventually in $q_3$.
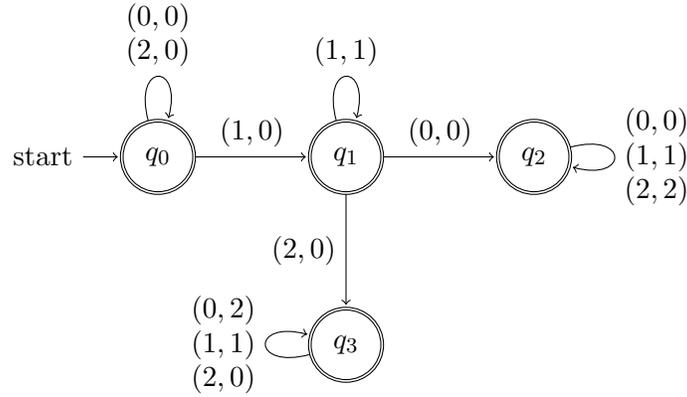
\begin{figure}[t]
\begin{tikzpicture}[node distance=2.5cm,on grid,auto, shorten > = 1pt]
	\node[state,initial,accepting](q0){$q_0$};
    \node[state,accepting](q1)[right=of q0]{$q_1$};
	\node[state,accepting](q2)[right=of q1]{$q_2$};
    \node[state,accepting](q3)[below=of q1]{$q_3$};
    \path[->]
    (q0) edge [loop above] node [align=center] {$(0,0)$\\$(2,0)$}()
   	edge node {$(1,0)$} (q1)
    (q1) edge [loop above] node [align=center] {$(1,1)$}()
    	edge node {$(0,0)$} (q2)
    	edge node [swap] {$(2,0)$} (q3)
    (q2) edge [loop right] node [align=center] {$(0,0)$\\$(1,1)$\\$(2,2)$}()
   (q3) edge [loop left] node [align=center] {$(0,2)$\\$(1,1)$\\$(2,0)$}();
\end{tikzpicture}
\caption{A B\"uchi automaton that $3$-accepts the graph of $d$.}%
\label{figure:distancecantor2}
\end{figure}
\end{exa}
\noindent A function $f : [0,1] \to \mathbb{R}$ is $\lambda$-\textbf{Lipschitz} if
\[
 | f(x) - f(y) | \leq \lambda | x - y | \quad \text{for all } 0 \leq x,y \leq 1,
 \]
and $f$ is Lipschitz if it is $\lambda$-Lipschitz for some $\lambda > 0$.
We say that $f$ is \textbf{H\"older continuous} with exponent $\gamma > 0$ if there is a $\lambda > 0$ such that
\[
 | f(x) - f(y) | \leq \lambda | x - y |^\gamma \quad \text{for all } 0 \leq x,y \leq 1.
\]
Note that Lipschitz functions are precisely functions which are H\"older continuous with exponent one.
It is also easy to see that any function which is H\"older continuous with exponent $\gamma > 1$ is constant.

\begin{factC}[{\cite[Theorem 10]{CSV}}]\label{fact:csv} Let $s \in \N_{>2}$, and suppose $f: [0,1]\to [0,1]$ is $(r,s)$-regular and continuous. Then $f$ is H\"older continuous with exponent $\log_{r}(s)$.
In particular:
\begin{enumerate}
\item If $r < s$, then $f$ is constant, and
\item if $r=s$, then $f$ is Lipschitz.
\end{enumerate}
\end{factC}

\noindent Fact~\ref{fact:csv} is proved in~\cite{CSV}, see~\cite[Theorem 10]{CSV}.
Indeed, it is claimed that $f$ is Lipschitz even when $r\neq s$, but the proof only treats the case $r=s=2$. An inspection of the proof however shows that when $r<s$, only the weaker condition of H\"older continuity is obtained. We now give an example of a $(4,2)$-regular function that is continuous, but not Lipschitz.\footnote{We thank Erin Caulfield for pointing out that this is a continuous regular function that is not Lipschitz.}


\begin{exa}
Let $h=(h_1,h_2) : [0,1]\to {[0,1]}^2$ be the Hilbert curve~\cite{Hilbert}. It can be shown that $h$ is $(4,2,2)$-regular and its graph is $(4,2,2)$-recognized by the automaton in Figure~\ref{fig:hilbert}. Since the graph of $h_1$ and the graph of $h_2$ are the projections of the graph of $h$ onto the first and third coordinates and onto the second and third coordinates, we have that $h_1$ and $h_2$ are $(4,2)$-regular. Observe that $h$ is not Lipschitz, as the image of a Lipschitz function cannot have higher Hausdorff dimension than the domain. Thus $h_1$ and $h_2$ cannot possibly be both Lipschitz, as this would imply Lipschitzness of $h$.
\end{exa}

\begin{figure}[t]
\begin{tikzpicture}[shorten >=1pt,node distance=3cm,on grid,auto,line width=1pt]
   \node[state,initial,accepting,initial where=right,initial distance=2ex] (q_0)   {$q_0$};
   \node[state,accepting] (q_1) [above right=of q_0] {$q_1$};
   \node[state,accepting] (q_2) [below right=of q_0] {$q_2$};
   \node[state,accepting](q_3) [below right=of q_1] {$q_3$};
    \path[->]
	(q_0) edge  node {(3,1,0)} (q_1)
		  edge  node [swap] {(0,0,0)} (q_2)
		  edge [loop left] node [align=center] {(1,0,1)\\(2,1,1)}()
	(q_1) edge  node {} (q_0)
		  edge  node  {(0,1,1)} (q_3)
		  edge [loop above] node {(2,0,0),(1,0,1)}()
	(q_3) edge  node {} (q_1)
		  edge  node {(3,0,1)} (q_2)
		  edge [loop right] node [align=center]{(2,0,0)\\(1,1,0)}()
	(q_2) edge  node {} (q_0)
		  edge  node {} (q_3)
		  edge [loop below] node {(2,1,1),(1,1,0)}();
\end{tikzpicture}
\caption{An automaton $(4,2,2)$-recognizing the Hilbert curve.}%
\label{fig:hilbert}
\end{figure}
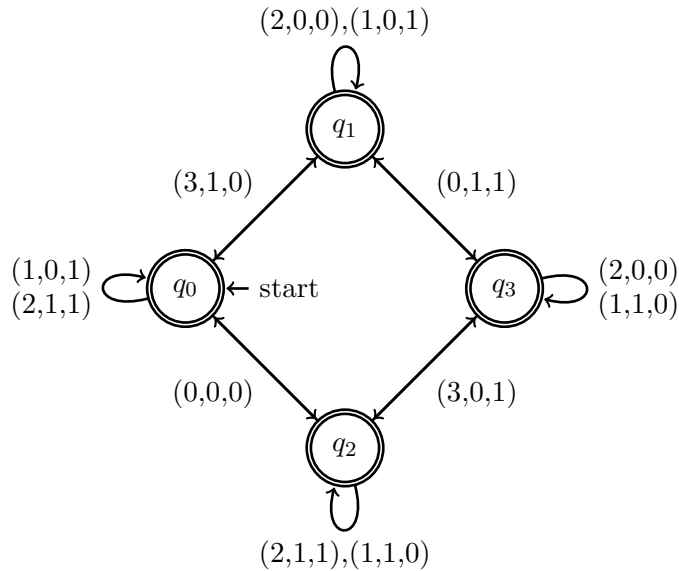

\subsection{Affine regular functions}\label{section:affine}
In this subsection we show that every affine $r$-regular function $J \to [0,1]$ defined on an interval $J \subseteq [0,1]$ is $\mathbb{Q}$-affine.
We require the following well-known fact.

\begin{fact}\label{fact:periodic}
Every B\"uchi automaton that accepts an infinite word, accepts an infinite periodic word.
\end{fact}

\noindent Lemma~\ref{lem:affine0} below is a corollary to Fact~\ref{fact:periodic}.
One can use Lemma~\ref{lem:affine0} to show that many familiar non-affine functions $[0,1] \to [0,1]$ are not $r$-regular.
For example, it is relatively easy to show that a polynomial of degree $\geq 2$ cannot be $r$-regular using Lemma~\ref{lem:affine0}.

\begin{lem}\label{lem:affine0}
Let $J$ be a nonempty subinterval of $[0,1]$ and let $f : J \to [0,1]$ be $r$-regular. Then $f(q) \in \mathbb{Q}$ for all $q \in \mathbb{Q} \cap J$.
\end{lem}

\noindent The proof below shows that Lemma~\ref{lem:affine0} holds more generally for $(r,s)$-regular functions.

\begin{proof}
By Fact~\ref{fact:periodic} every $r$-recognizable subset of ${[0,1]}^n$ contains a point with rational coordinates.
In particular, every $r$-recognizable singleton has rational coordinates.
Fix a rational $0 \leq q \leq 1$.
It is easy to see that $\{ (q,y) : 0 \leq y \leq 1 \}$ is $r$-recognizable.
As $\Gamma(f)$ is $r$-recognizable, and the intersection of two $r$-recognizable sets is $r$-recognizable, it follows that
\[
 \{(q,f(q))\} = \{ (q,y) : 0 \leq y \leq 1 \} \cap \Gamma(f)
 \]
is $r$-recognizable. Thus $f(q)$ is rational.
\end{proof}

\noindent Lemma~\ref{lem:affine0} and the fact that an affine function that takes rational values on rational points is necessarily $\mathbb{Q}$-affine, together imply the following:

\begin{prop}\label{prop:qaffine}
Let $J \subseteq [0,1]$ be a nonempty open interval and let $f : J \to [0,1]$ be $r$-regular and affine. Then $f$ is $\mathbb{Q}$-affine.
\end{prop}

\noindent Proposition~\ref{prop:qaffine} is a special case of a result on first-order expansions $(\mathbb{R},<,+,0)$, see~\cite[Theorem 3.9]{FHW-Compact}.
This more general result shows that Proposition~\ref{prop:qaffine} still holds when the graph of $f$ is $r$-recognized by an B\"uchi automaton with advice, where the automaton that recognizes the graph can access a fixed infinite advice string.

\subsection{Graph Directed Iterated Function Systems and Self-similar sets} Let $(X, d)$ and $(X', d')$ be metric spaces.
A \textbf{similarity of ratio $r$} between $(X,d)$ and $(X',d')$ is a function $f:X \to X'$ such that
\[
d'(f(x), f(y)) = \frac{1}{r}d(x,y) \quad \text{for all} \quad x,y \in X.
\]
A \textbf{graph-directed iterated function system (of ratio $r$)}, or a GDIFS for short, is a quadruple $\calG = (V,E,s,t,X,S)$ such that
\begin{enumerate}
\item $(V,E,s,t)$ is a graph,
\item $X = {( X_v,d_v )}_{v \in V }$ is a family of metric spaces, and
\item $S = {(S_e : X_{t(e)} \to X_{s(e)})}_{e \in E }$ is a family of functions such that $S_e$ is a similarity of ratio $r$ for all $e \in E$.
\end{enumerate}
We typically drop the source and range maps.
We let $E_{u,v}$ denote the set of edges $e \in E$ from $u$ to $v$.
An \textbf{attractor} for $\calG$ is a collection $K = {(K_u)}_{u \in V}$ such that each $K_u$ is a non-empty compact subset of $X_u$ and
\[
K_u = \bigcup_{v \in V} \bigcup_{e \in E_{u,v}} S_e(K_v) \quad \text{for all } u \in V.
\]

\noindent Let $\calA=(Q,\Sigma_r^n,\Delta,I,F)$ be a B\"uchi automaton. We associate  a GDIFS $\calG_{\calA}$ of ratio $r$ to $\calA$ as follows.
Let $V := Q$.
Let $(X_v,d_v)$ be ${[0,1]}^n$ with the usual euclidean metric for each $v \in V$.
Set $E := \Delta$, and for each $\delta=(u,\sigma,v) \in \Delta$ let
\[
S_\delta(x) = \frac{ x + \sigma }{ r } \quad \text{for all } x \in {[0,1]}^n.
\]

\noindent All GDFIS considered will be of the form $\calG_{\calA}$ for some B\"uchi automaton $\calA$ over $\Sigma_r^n$. While this is just a small subset of all GDFIS, we need the generality of GDFIS in contrast to just iterative function systems to account for the case when $\calA$ has more than a single state.

\medskip For $p,q \in Q$, let $W_{p,q}(n)$ be the set of paths of length $n$ from $p$ to $q$. Since ${(K_q)}_{q\in Q}$ is the attractor of $\calG_{\calA}$, we obtain the following fact by an iterative application of the attractor property.

\begin{fact}%
\label{fact:attr0}
Let $n\in N$ and $p \in Q$. Then
\[
 K_{p} = \bigcup_{q \in Q}\bigcup_{(\delta_1,\ldots,\delta_n) \in W_{p,q}(n)} (S_{\delta_1} \circ \ldots\circ S_{\delta_n})(K_q).
\]
\end{fact}

\noindent Throughout this paper will use the following result from~\cite{CLR} connecting attractors of GDFIS and $r$-recognizable sets.

\begin{factC}[{\cite[Theorem 57]{CLR}}]%
\label{fact:thmclr}
Suppose $\calA$ is trim and closed.  Then
\[
V_r(\calA)  = \bigcup_{q \in I} K_q,
\]
where ${(K_q)}_{q\in Q}$ is the unique attractor of $\calG_{\calA}$.
\end{factC}

\noindent We observe that the GDFIS $\calG_{\calA}$ only depends on $Q, n,r$ and $\Delta$, but not on $I$ or $F$. Therefore the unique attractor of $\calG_{\calA}$ also only depends on $Q, n,r$ and $\Delta$, but not on $I$ or $F$.
%

\subsection{Logic}
As shown in Boigelot, Rassart, and Wolper~\cite{BRW}, there is also a connection between $r$-recognizable sets and sets definable in a particular first-order expansion of the ordered real additive group $(\R,<,+,0)$.
 Let $V_r(x,u,k)$ be the ternary predicate on $\R$ that holds whenever $u = r^n$ for $n \in \mathbb{N}, n \geq 1$ and there is a base-$r$ expansion of $x$ with $n$-th digit $k$.
Let $\mathscr{T}_r$ be the first-order structure $(\mathbb{R},<,+,0,V_r)$.

\begin{factC}[{\cite[Theorem 5]{BRW}}]%
\label{BRW}
Let $X\subseteq {[0,1]}^n$. Then $X$ is $r$-recognizable if and only if $X$ is definable in $\mathscr{T}_r$ (without parameters).
\end{factC}

\noindent Fact~\ref{BRW} is a powerful tool for establishing that the collection of $r$-recognizable sets is closed under various operations.
For example it may be used to easily show that the interior and closure of an $r$-recognizable set are also $r$-recognizable.

\section{Tools from metric geometry}

\subsection{Hausdorff Measure}\label{section:hausdorff}
We recall the definition of the one-dimensional Hausdorff measure (generalized arc length) of a subset $A$ of $\mathbb{R}^n$.
We let $\Diam(B) \in \R \cup \{ \infty \}$ be the diameter of a subset $B$ of $\R^n$. Given $\delta > 0$
 we denote by $\mathcal{C}_{\delta}(X)$ the collection of all countable collections $\mathcal{B}$ of closed balls of diameter $\leq \delta$ covering $X$.
Given $\delta >0$ we declare
\[
\mathcal{H}^{1}_\delta(X) = \inf_{\mathcal{B}\in \mathcal{C}_{\delta}} \sum_{B \in \mathcal{B}} \Diam(B) \in \mathbb{R} \cup \{\infty\}.
\]
Note that $\mathcal{H}^{1}_\delta(X)$ decreases with $\delta$. Thus  either $\lim_{\delta \rightarrow 0} \mathcal{H}^{1}_\delta(X)$ or is infinite. The one-dimensional Hausdorff measure of $X$ is
\[
\mathcal{H}^{1}(X):= \sup_{\delta>0} H^1_{\delta}(X) = \lim_{\delta \rightarrow 0} \mathcal{H}^{1}_\delta(X).
\]
Recall that $\mathcal{H}^1$ agrees with the usual Lebesgue measure on $[0,1]$.
We now prove a general lemma about graphs of Lipschitz functions.

\begin{lem}\label{lem:hc} Let $f : [0,1]\to [0,1]$ be a Lipschitz function and let $X$ be a connected subset of $\Gamma(f)$.
Then
\[
\calH^1(X) \geq \Diam(X).
\]
\end{lem}
\begin{proof}
Set $Y:=\pi(X)$. Since $X$ is connected, so is $Y$.
Since $f|_Y$ is a Lipschitz function and $Y$ is connected, $\calH^1(\Gamma(f|_Y))$ is equal to the arc length of $\Gamma(f|_Y)$.
Thus $\calH^1(\Gamma(f|_Y)) \geq \Diam(\Gamma(f|_Y))$.
Now observe that $X = \Gamma(f|_{Y})$.
 \end{proof}

\subsection{Nowhere dense recognizable sets are Lebesgue null}\label{section:null}
We first recall an important definition from metric geometry.
A subset $X$ of $\mathbb{R}^n$ is \textbf{$\varepsilon$-porous} if every open ball in $\mathbb{R}^n$ with radius $t$ contains an open ball of radius $\varepsilon t$ that is disjoint from $X$.
A subset of $\mathbb{R}^n$ is \textbf{porous} if it is $\varepsilon$-porous for some $\varepsilon > 0$.
Lemma~\ref{lem:porous} holds for recognizable subsets of ${[0,1]}^n$.
We only require the one-dimensional case, so to avoid very mild technicalities we only treat that case.

\begin{lem}\label{lem:porous}
Let $X \subseteq [0,1]$ be $r$-recognizable and nowhere dense.
Then $X$ is porous.
\end{lem}

\noindent In the following proof an $r^{-n}$-\textbf{interval} is an interval of the form $\left[ \frac{m}{r^n}, \frac{m+1}{r^n} \right]$.

\begin{proof}
Recall that the $r$-kernel of $X$ is the collection of sets of the form
\[
r^n\left( X \cap \left[ \frac{m}{r^n}, \frac{m+1}{r^n} \right] - m \right),
\]
where $n \geq 1$ and $0 \leq m \leq r^n - 1$.
By Fact~\ref{fact:similar}, the $r$-kernel of $X$ consists of only finitely many distinct subsets of $[0,1]$, each of which has empty interior.
Thus, for some $k \geq 1$, there is an $r^{-k}$-interval that is disjoint from every element of the $r$-kernel of $X$.
It follows that any $r^{-n}$-interval contains an $r^{-(n + k)}$-interval which is disjoint from $X$.
Note that any interval $(x - t, x + t) \subseteq [0,1]$ contains an $r^{-n}$-interval of length $\geq r^{-1} 2t$.
Thus any interval $(x - t, x + t) \subseteq [0,1]$ contains an $r^{-n}$-interval with length $\geq r^{-(k+1)} 2t$ that is disjoint from $X$.
Thus $X$ is $2r^{-(k+1)}$-porous.
\end{proof}

\noindent It is shown in Luukkainen~\cite[Theorem 5.2]{Luukk} that a subset $X$ of $\mathbb{R}^n$ is porous if and only if the Assouad dimension of $X$ is strictly less than $n$.
The Hausdorff dimension of a closed subset of $\mathbb{R}^n$ is bounded above by its Assouad dimension.
It follows that a nowhere dense $r$-recognizable subset of ${[0,1]}^n$ has Hausdorff dimension $< n$.
In particular a nowhere dense $r$-recognizable subset of $[0,1]$ is Lebesgue null.
For the sake of completeness we provide the following more elementary proof.

\begin{prop}\label{prop:porous}
Every porous subset of $[0,1]$ is Lebesgue null. In particular, every nowhere dense $r$-recognizable subset of $[0,1]$ is Lebesgue null.
\end{prop}

\begin{proof}
Let $X \subseteq [0,1]$ be $\varepsilon$-porous.
Note that the closure of $X$ is also $\varepsilon$-porous.
After replacing $X$ by its closure if necessary we suppose $X$ is closed, and so in particular, Lebesgue measurable.
Every interval $(x-t,x+t) \subseteq [0,1]$ contains an interval with length at least $2\varepsilon t$ that is disjoint from $X$.
Hence
\[
 \calH^1( (x - t, x +t)  \cap X) \leq (1 - \varepsilon) 2t.
\]
Thus
\[
\frac{\calH^1( (x - t, x +t)  \cap X)}{2t} \leq 1 - \varepsilon
\]
for all $x,t > 0$ such that $(x-t,x+t) \subseteq [0,1]$.
Suppose $X$ has positive Lebesgue measure.
By the Lebesgue density theorem we have
\[
\lim_{t \to 0} \frac{ \calH^1( (x - t ,x + t) \cap X)}{2t} = 1
\]
for $\calH^1$-almost all $x \in X$.
This is a contradiction.
\end{proof}

\section{The strongly connected case}

In this section we prove Theorem A in the special case when the $r$-regular function under consideration is accepted by a strongly connected determinstic B\"uchi automaton. Indeed, we establish the following strengthening of Theorem A in this case.

\begin{thm}\label{prop:affine}
Let $f: [0,1] \to [0,1]$ be an $r$-regular continuous function whose graph is $r$-accepted by a strongly connected deterministic B\"uchi automaton $\calA$. Then $f$ is $\Q$-affine.
\end{thm}
\noindent This extends a classical result of Hutchinson~\cite[Remark 3.4]{Hutch} for functions whose graph is the attractor of an iterative function system. The Hutchinson's result almost directly implies Theorem~\ref{prop:affine} when the automaton $\calA$ only has one state.

\subsection{Two lemmas}
We collect two lemmas about strongly connected B\"uchi automata before proving Theorem~\ref{prop:affine}.
We let $\calA = (Q,\Sigma,\Delta,I,F)$ be a B\"uchi automaton.
Recall that for $p,q\in Q$, the set $P_{p,q}(n)\subseteq \Sigma^n$ is the set of all labels of paths of length $n$ from $p$ to $q$.

\begin{lem}\label{enoughpaths}
Suppose $\calA$ is strongly connected. Let $p \in Q$ be such that
\begin{equation*}\label{eq:enoughpaths}
\left|\bigcup_{q \in Q} P_{p,q}(n)\right| \geq r^n \quad \text{for all } n \in \N.
\end{equation*}
Then
\[
\limsup\limits_{n\to \infty}\frac{|P_{p,q}(n)|}{r^n} >0 \quad \text{for all } q \in Q.
\]
\end{lem}
\begin{proof}
Since $Q$ is finite, the pigeonhole principle yields a state $p_0\in Q$ such that
\[
\limsup\limits_{n\to \infty}\frac{|P_{p,p_0}(n)|}{r^n} >0.
\]
Let $q\in Q$. Since $\calA$ is strongly connected, there is a path of length $k$ from $p_0$ to $q$. Thus  for all $n\in \N$
\[
|P_{p,q}(n+k)|\geq |P_{p,p_0}(n)|.
\] Therefore
\[
\limsup\limits_{n\to \infty}\frac{|P_{p,q}(n)|}{r^n} = \limsup\limits_{n\to \infty}\frac{|P_{p,q}(n+k)|}{r^{n+k}} \geq \frac{1}{r^k}\limsup\limits_{n\to \infty}\frac{|P_{p,p_0}(n)|}{r^n} >0.\qedhere
\]
\end{proof}

\begin{lem}\label{lem:stint1} Suppose $\calA$ is strongly connected, deterministic and $\Sigma = \Sigma_r$. If the interior of $V_r(\calA)$ is non-empty, then $V_r(\calA)=[0,1]$.
\end{lem}
\begin{proof}
Set $X := V_r(\calA)$.
Let $J\subseteq [0,1]$ be an open interval contained in $X$.
Towards a contradiction, we suppose that $[0,1] \setminus X$ is non-empty. Fix $z\in [0,1]\setminus X$.
Let $w_1,w_2\in \Sigma_r^{\omega}$ be such that
\[
\{ w \in \Sigma_r^{\omega} \ : \ v_r(w) = z \} = \{w_1, w_2\}.
\]
The main idea is as follows: By strong connectedness of $\calA$ we can preceed either $w_1$ or $w_2$ with any prefix that is consistent with being in $X$, and still land outside of $X$. This will contradict $J\subseteq X$. We now make this argument precise.

\medskip
Since $z\notin X$, there is a state $p_0\in Q$ such that there is no run of either $w_1$ or $w_2$ from $p_0$. Since $\calA$ is strongly connected, there are $t\in \N$, $u\in \Sigma_r^t$ and a run $s_0s_1\dots s_t$ of $u$ such that $s_t=p_0$ and
\begin{equation}\label{lem:stint}
v_r(\{ uw \ : \ w \in \Sigma_r^{\omega}\}) \subseteq J.
\end{equation}
Since $\calA$ is deterministic, our assumption on $w_1$ and $w_2$ implies that the set $\{ w \in \Sigma_r \ : \ uw \in L(\calA) \}$ contains neither $w_1$ nor $w_2$. Set $y:=v_r(uw_1)$. Since $v_r(uw_1)=v_r(uw_2)$, we have that
\[
\{ w \in \Sigma_r^{\omega} \ : \ v_r(w) = y \} = \{uw_1, uw_2\}.
\]
Thus $y \notin X$. Together with~\eqref{lem:stint}, this contradicts $J\subseteq X$.
\end{proof}

\noindent Recall that $\pi : \R^2 \to \R$ is the coordinate project onto the first coordinate.

\begin{cor}\label{cor:lemint} Suppose $\calA$ is strongly connected, weak, deterministic and $\Sigma=\Sigma_r^2$. If the interior of $\pi(V_r(\calA))$ is non-empty, then $\pi(V_r(\calA))=[0,1]$.
\end{cor}
\begin{proof}
Assume that the interior of $\pi(V_r(\calA))$ is non-empty.
Since $\calA$ is weak and strongly connected and $L(\calA)$ is non-empty, all states in $\calA$ are final.
Let
\[
U:= \{ u_1 \in \Sigma_r^{\omega} \ : \ \exists u_2 \in \Sigma_r^{\omega} \ u_1\times u_2 \in L(\calA)\}.
\]
Note that $v_r(U) = \pi(V_r(\calA))$. It is easy to see that there is a strongly connected B\"uchi automaton $\calA'$ over $\Sigma_r$ such that $L(\calA')=U$. Since all states of $\calA$ are final, we can choose $\calA'$ such that all states of $\calA'$ are final as well. Because of this, we can use the usual powerset construction to obtain a strongly connected \emph{deterministic} B\"uchi automaton $\calA''$ over $\Sigma_r$ such that $L(\calA')=L(\calA'')$. The corollary now follows from Lemma~\ref{lem:stint1}.
\end{proof}

\begin{cor}\label{cor:interiorallstates} Suppose $\calA$ is strongly connected, weak, deterministic and $\Sigma=\Sigma_r^2$. Let $p,p' \in Q$ and let ${(K_q)}_{q \in Q}$ be the attractor of the GDFIS $\calG_{\calA}$. Then $\pi(K_p) = [0,1]$ if and only if $\pi(K_{p'})=[0,1]$.
\end{cor}
\begin{proof}
Set $\calA_{p}:=(Q,\Sigma_r^2,\Delta,\{p\},F)$ and set $\calA_{p'}:=(Q,\Sigma_r^2,\Delta,\{p'\},F)$. Since the GDFIS $\calG_{\calA}$ does not dependent on the initial states of $\calA$, we obtain from Fact~\ref{fact:thmclr} that
\[
V_r(\calA_p) = K_p \hbox{ and } V_r(\calA_{p'}) = K_{p'}.
\]
By Corollary~\ref{cor:lemint} it is left to show that the interior of $\pi(K_p)$ is non-empty if and only if the interior of $\pi(K_{p'})$ is non-empty.
Suppose that $\pi(K_{p'})$ has non-empty interior. Since $\calA$ is strongly connected, there is $n\in \N$ such that there is a path $(\delta_1,\dots,\delta_n)$ of length $n$ from $p$ to $p'$. Thus by Fact~\ref{fact:attr0}
\begin{equation}\label{eq:cor45}
(S_{\delta_1} \circ \dots \circ S_{\delta_n})(K_{p'}) \subseteq K_p.
\end{equation}
It is easy to see that since the interior $\pi(K_{p'})$ is non-empty, so is the interior of $\pi\big((S_{\delta_1} \circ \dots \circ S_{\delta_n})(K_{p'})\big)$. By~\eqref{eq:cor45} the interior of $\pi(K_p)$ is non-empty.
\end{proof}

\subsection{Proof of Theorem~\ref{prop:affine}}
Throughout this subsection, fix an $r$-regular continuous $f: [0,1]\to[0,1]$ whose graph is $r$-accepted by a strongly connected deterministic B\"uchi automaton $\calA=(Q,\Sigma_r,\Delta,\{p\},F)$.  As $\calA$ is strongly connected, it is trim.
As $\Gamma(f)$ is closed, we may assume that $\calA$ is closed.
Let  $\calG_{\calA} = (V,E,X,S)$ be the GDIFS associated to $\calA$. Let ${(K_q)}_{q\in Q}$ be the attractor of the GDIFS $\calG_{\calA}$. By Fact~\ref{fact:csv}, $\Gamma(f) = K_p$.

\medskip
Recall that for $q,q' \in Q$, we denote by $W_{q,q'}(n)$ be the set of paths of length $n$ from $q$ to $q'$. Since ${(K_q)}_{q\in Q}$ is the attractor of $\calG_{\calA}$, we have  by Fact~\ref{fact:attr0} that for each $n\in \N$
\begin{equation}\label{eq:attr0}
\Gamma(f) = K_{p} = \bigcup_{q \in Q}\bigcup_{(\delta_1,\ldots,\delta_n) \in W_{p,q}(n)} (S_{\delta_1} \circ \ldots\circ S_{\delta_n})(K_q).
\end{equation}
\noindent For $q,q' \in Q$,
each path $(\delta_1,\ldots,\delta_n)$ of length $n$ from $q$ to $q'$, is uniquely determined from its label $\sigma = \sigma_1\ldots\sigma_n \in P_{q,q'}(n)$, since $\calA$ is deterministic. We let $\xi_\sigma:X_{q'} \to X_q$ be the composition $\xi_\sigma:=S_{\delta_1} \circ \ldots\circ S_{\delta_n}$.
Then
\[
\xi_\sigma(x) =\frac{x+\sum_{i=1}^{n} \sigma_{i} r^{n-i}}{r^n}
\]
With this new notation, equation~\eqref{eq:attr0} becomes
\begin{equation}\label{eq:attr}
\Gamma(f) = K_{p} = \bigcup_{q \in Q}\bigcup_{\sigma \in P_{{p},q}(n)} \xi_\sigma(K_q).
\end{equation}

\begin{lem}\label{disjoint}
Let $q,q' \in Q$, and let $\sigma\in P_{p,q}(n)$ and $\sigma' \in P_{p,q'}(n)$ be such that $\sigma\neq \sigma'$. Then
\[
\calH^1\big(\xi_\sigma(K_q)\cap \xi_{\sigma'}(K_{q'})\big) = 0.
\]
\end{lem}
\begin{proof}
The main difficulty is to show that the overlap of the boundaries of the two sets has zero measure. Set $a := (a_1,a_2) := \sum_{i=1}^{n} \sigma_{i} r^{n-i}$ and $b:=(b_1,b_2):=\sum_{i=1}^{n} \sigma'_{i} r^{n-i}$.
Thus
\[
\xi_\sigma(x) =\frac{x+a}{r^n} \ \hbox{ and } \ \xi_{\sigma'}(x) =\frac{x+b}{r^n}.
\]
Observe that $\max\{a_1,a_2,b_1,b_2\}<r^{n}$.
We get that
\[
\xi_\sigma(K_q) \subseteq \xi_\sigma\left({[0,1]}^2\right) = \left[\frac{a_1}{r^n},\frac{a_1+1}{r^n}\right]\times \left[\frac{a_{2}}{r^n},\frac{a_{2}+1}{r^n}\right]
\]
and
\[
\xi_{\sigma'}(K_{q'}) \subseteq \xi_{\sigma'}\left({[0,1]}^2\right) = \left[\frac{b_1}{r^n},\frac{b_1+1}{r^n}\right]\times \left[\frac{b_{2}}{r^n},\frac{b_{2}+1}{r^n}\right].
\]
Since $\xi_\sigma\neq \xi_{\sigma'}$, there is some $i \in \{1,2\}$ such that $a_i\neq b_i$. Observe that when both $a_1\neq b_1$ and $a_2\neq b_2$, the statement of the Lemma follows immediately.
We therefore can reduce to the case when there is $i \in \{1,2\}$ such that $a_i= b_i$.

\medskip Suppose that $a_2 = b_2$. Without loss of generality we can assume $a_1<b_1$. Then
\[
\xi_\sigma(K_q) \cap \xi_{\sigma'}(K_{q'}) \subseteq \left\{\frac{a_1+1}{r^n}\right\}\times \left[\frac{a_2}{r^n},\frac{a_2+1}{r^n}\right].
\]
By~\eqref{eq:attr}, we have $\xi_\sigma(K_q) \cap \xi_{\sigma'}(K_{q'}) \subseteq \Gamma(f)$. This directly implies that $\xi_\sigma(K_q) \cap \xi_{\sigma'}(K_{q'})$ must be a singleton as $f$ is a function.

\medskip
Suppose that $a_1=b_1$. Without loss of generality we can assume $a_2 < b_2$. Then
\[
\xi_\sigma(K_q) \cap \xi_{\sigma'}(K_q) \subseteq \left[\frac{a_1}{r^n},\frac{a_1+1}{r^n}\right] \times \left\{\frac{a_2+1}{r^n}\right\}.
\]
Towards a contradiction, suppose that  $\calH^1\left(\xi_\sigma(K_q)\cap \xi_{\sigma'}(K_{q'})\right) > 0$. It is easy to check that $\xi_\sigma(K_q)\cap \xi_{\sigma'}(K_{q'})$ is $r$-recognizable.
Thus Lemma~\ref{lem:porous} and Proposition~\ref{prop:porous} yield a non-empty open interval $J\subseteq \left[\frac{a_1}{r^n},\frac{a_1+1}{r^n}\right]$ such that
\[
J \times  \left\{\frac{a_2+1}{r^n}\right\} \subseteq \xi_\sigma(K_q) \cap \xi_{\sigma'}(K_{q'}).
\]
Since  $\xi_\sigma(K_q) \cap \xi_{\sigma'}(K_{q'}) \subseteq \Gamma(f)$ by~\eqref{eq:attr}, this shows $f$ is constant and equal to $\frac{a_2+1}{r^n}$ on the interval $J$.
Since $\calA$ is strongly connected, there is $\ell \in \N_{>0}$ and $\nu=\nu_1\dots\nu_{\ell} \in P_{q,p}(\ell)$ such that
\begin{equation}\label{eq:lemmaJ}
\xi_{\sigma\nu}({[0,1]}^2) \subseteq J \times \left[\frac{a_2}{r^n},\frac{a_2+1}{r^n}\right].
\end{equation}
Set $c:=(c_1,c_2):= \sum_{i=1}^{\ell} \nu_{i} r^{\ell-i}$. Note that $c_1<r^{\ell}$ and $c_2 < r^{\ell}$. Since $a_2<b_2 <r^{n}$, we have
\[
a_2 + c_2r^n< b_2 + c_2 r^n < r^{n+\ell}.
\]
Let $x=(x_1,x_2) \in {[0,1]}^2$.
Then
\begin{align*}
\xi_{\sigma\nu\sigma}(x) 
&=\frac{x + a r^{n+\ell} + c r^n + a}{r^{2n+\ell}}.
\end{align*}
Set $y=(y_1,y_2):= \xi_{\sigma\nu\sigma}(x)$. By~\eqref{eq:lemmaJ} we have $y_1\in J$. Furthermore,
\[
y_2 = \frac{x_2 + a_2 r^{n+\ell} + c_2 r^n + a_2}{r^{2n+\ell}}
\leq \frac{1 + a_2 r^{n+\ell} + c_2 r^n + a_2}{r^{2n+\ell}}
=\frac{a_2}{r^n} +
\frac{1}{r^{n}}\frac{c_2 r^{n}+ a_2+1}{r^{n+\ell}} < \frac{a_2+1}{r^n}.
\]
Now suppose $x \in K_q$. Then $y \in \Gamma(f)$. However, since $y_1 \in J$, we get $y_2 =\frac{a_2+1}{r^n}$. This contradicts $y_2 < \frac{a_2+1}{r^n}$.
\end{proof}

\begin{proof}[Proof of Theorem~\ref{prop:affine}]

By~\eqref{eq:attr}, $K_p=\Gamma(f)$. Thus $K_p$ is compact and connected. By Lemma~\ref{lem:hc} we have $\calH^1(K_p) \geq \Diam(K_p)$. As in the proof of~\cite[Remark 3.4]{Hutch}, we observe that in order to show that $K_p$ is a line segment, it is enough to establish $\calH^{1}(K_p) = \Diam(K_p)$. Thus we will now show that $\calH^1(K_p) \leq \Diam(K_p)$.


\medskip
Let $q \in Q$, and let $n\in \N$ be such that there is $\sigma \in P_{{p},q}(n)$.  Observe that
$\calH^1(\xi_\sigma(K_q))= r^{-n}  \calH^1(K_q)$. By~\eqref{eq:attr} we have that $\xi_{\sigma}(K_q)\subseteq K_p = \Gamma(f)$. Since $\pi(K_p)=[0,1]$, Corollary~\ref{cor:interiorallstates} gives that $\pi(K_q)=[0,1]$. Thus $\pi(\xi_{\sigma}(K_q))$ is connected. Since $\xi_{\sigma}(K_q)\subseteq \Gamma(f)$,
we deduce that $\xi_{\sigma}(K_q)$ is connected as well. Therefore we obtain $\calH^1(\xi_\sigma(K_q))\geq\Diam(\xi_\sigma(K_q))$ by Lemma~\ref{lem:hc}. Thus
\begin{equation}\label{eq:hutlemma}
\calH^1(K_q)= r^{n}  \calH^1(\xi_\sigma(K_q))\geq r^n\Diam(\xi_\sigma(K_q)) = \Diam(K_q).
\end{equation}
Since any overlap in the union in~\eqref{eq:attr} has zero 1-dimensional Hausdorff measure by Lemma~\ref{disjoint}, we can compute the 1-dimensional Hausdorff measure of $K_p$ using~\eqref{eq:attr} as follows:
\begin{equation}\label{eq:hausdorffmeasure}
\calH^1(K_{p}) = \sum_{q \in Q}\sum_{\sigma \in P_{p,q}(n)} \frac{1}{r^n}\calH^1(K_q) = \sum_{q\in Q}\frac{|P_{p,q}(n)|}{r^n}\calH^1(K_q).
\end{equation}


\noindent Suppose $\calH^1(K_{p}) > \Diam(K_{p})$ towards a contradiction. Set $\varepsilon := \calH^1(K_{p}) -  \Diam(K_{p})>0$. For $n\in \N$ and $\sigma \in P_{p,q}(n)$, we note that
\[
\Diam(\xi_{\sigma}(K_q))=\frac{1}{r^n} \Diam(K_q).
\]
Thus by~\eqref{eq:attr}
\[
\calH^1_{r^{-n}}(K_p) \leq \sum_{q\in Q}\frac{|P_{p,q}(n)|}{r^n} \Diam(K_q).
\]
By the definition of Hausdorff measure
\[
\calH^1(K_{p}) \leq \limsup\limits_{m \to \infty} \sum_{q\in Q}\frac{|P_{p,q}(m)|}{r^m}\Diam(K_q).
\]
However, for every $n\in \N$, we obtain using~\eqref{eq:hausdorffmeasure} and~\eqref{eq:hutlemma}
\begin{align*}
\calH^1(K_{p})&= \frac{|P_{p,p}(n)|}{r^n}\calH^1(K_{p})+\sum_{q\in Q\setminus \{p\}}\frac{|P_{p,q}(n)|}{r^n}\calH^1(K_q)\\
&\geq\frac{|P_{p,p}(n)|}{r^n}\varepsilon+\sum_{q\in Q}\frac{|P_{p,q}(n)|}{r^n} \Diam(K_q).
\end{align*}
To reach a contradiction, it is left to prove that $\limsup\limits_{n \to \infty}\frac{|P_{p,p}(n)|}{r^n}$ is positive.
Since $f$ is a function with domain $[0,1]$, we know that $\pi(K_p)=[0,1]$. By~\eqref{eq:attr} this means that $[0,1]$ is covered by $\bigcup_{q \in Q}\bigcup_{\sigma \in P_{{p},q}(n)} \pi(\xi_\sigma(K_q))$. Since $\Diam(\pi(\xi_{\sigma}(K_q)))\leq r^{-n}$ for each $\sigma\in P_{p,q}(n)$, we have for each $n\in \N$
\[
\left|\bigcup_{q \in Q} P_{p,q}(n)\right| \geq r^n.
\]
Thus Lemma~\ref{enoughpaths} shows $\limsup\limits_{n \to \infty}\frac{|P_{p,p}(n)|}{r^n}$ is positive.
\end{proof}

\section{Proof of Theorem A}

In this section we will give the proof of Theorem A. Indeed, we will establish the following strengthening of Theorem A.

\begin{thm}\label{nowhere}
Let $f: [0,1] \to [0,1]$ be a continuous $r$--regular function. Then there are pairwise disjoint open $r$-recognizable subsets $U_1,\ldots,U_m$ of $[0,1]$ and rational numbers $\alpha_1,\ldots,\alpha_m$ such that
\begin{enumerate}
\item $[0,1]\setminus \bigcup_{i = 1}^{m} U_i$ is Lebesgue null, and
\item the restriction of $f$ to a connected component of $U_i$ is affine with slope $\alpha_i$.
\end{enumerate}
\end{thm}

\noindent We first observe that we can reduce to the case of functions $r$-accepted by certain well-behaved automata.

\begin{prop}\label{Th57}
Let $f:[0,1]\to [0,1]$ be an $r$-regular continuous function.
Then its graph $\Gamma(f)$ is $r$-accepted by a trim, deterministic, closed B\"uchi automaton.
\end{prop}
\begin{proof}
As $f$ is continuous, its graph $\Gamma(f)$ is closed in ${[0,1]}^2$.
Thus, there is a trim B\"uchi automaton $\calA$ $r$-accepting $\Gamma(f)$ which, by Fact~\ref{closed}, we may assume to be closed. A closed automaton is weak. Therefore there is a deterministic trim B\"uchi automaton $\calA'$ such that $L(\calA)=L(\calA')$. Again, we can take $\calA'$ to be closed by Fact~\ref{closed}.
\end{proof}

\noindent Throughout this section, let $\pi : \R^2 \to \R$ be the coordinate projection onto the first coordinate.

\subsection{Reduction to full automata}
Given a B\"uchi automaton $\calA =(Q,\Sigma,\Delta,I,F)$ and subsets $Q_1,Q_2$ of $Q$ with $Q_2\subseteq Q_1$, we let $\calA_{Q_1,Q_2}$ be the automaton
\[
(Q_1,\Sigma,\Delta \cap (Q_1 \times \Sigma \times Q_1),Q_2,F\cap Q_1).
\]


\begin{defi} Let $\calA$ be a B\"uchi automata. We say $\calA$ is \textbf{full} if  for every sink $Q'$ of $\calA$ there is $q \in Q'$ such that
\[
\pi(V_r(\calA_{Q',\{q\}})) = [0,1].
\]
\end{defi}

\begin{lem}\label{cor:connectedcomp} Let $f: [0,1] \to [0,1]$ be an  $r$-regular function. Then there is a closed full deterministic B\"uchi automaton $\calA$ that $r$-accepts $\Gamma(f)$.
\end{lem}
\begin{proof}
Let $\calA = (Q,\Sigma_r^2,\Delta,\{q_0\},F)$ be a closed deterministic B\"uchi automaton that $r$-accepts $\Gamma(f)$. We prove that there is a closed full deterministic B\"uchi automaton $\calA'$ such that $V_r(\calA')=\Gamma(f)$. We do so by induction on the number of strongly connected components of $\calA$.

\medskip
First suppose that $\calA$ is strongly connected. Then $Q$ is the only sink of $\calA$ and $\calA_{Q,\{q_0\}}=\calA$. Thus
\[
\Gamma(f) = V_r(\calA) = V_r(\calA_{Q,\{q_0\}}).
\]
Thus $\pi(V_r(\calA_{Q,\{q_0\}})) = [0,1]$ by Corollary~\ref{cor:lemint}.

\medskip
Now suppose $\calA$ has $n$ strongly connected components, where $n>1$. Let $Q_{0}$ be the set of states in the strongly connected component containing $q_0$. Set $X_0:=V_r(\calA_{Q_0,\{q_0\}})$.

\medskip
First suppose the interior of $\pi(X_0)$ is non-empty. Since $\calA_{Q_0,\{q_0\}}$ is strongly connected and closed, we obtain $\pi(X_0)=[0,1]$ by Corollary~\ref{cor:lemint}. Since $X_0 \subseteq \Gamma(f)$, we have $X_0=\Gamma(f)$. We declare $\calA':= \calA_{Q_0,\{q_0\}}$.

\medskip
We have reduced to the case that the interior of $\pi(X_0)$ is empty. In this case, $Q_0$ cannot be a sink. Suppose there is a sink $Q'\subseteq Q$ such that
\[
\pi(V_r(\calA_{Q',\{q\}}))\neq [0,1]
\]
for all $q\in Q'$. Since $\calA_{Q',\{q\}}$ is closed, we have that $\pi(V_r(\calA_{Q',\{q\}}))$ is closed and thus nowhere dense for all $q\in Q'$.
Observe that
\[
L(\calA_{Q',Q'}) = \bigcup_{q\in Q'} L(\calA_{Q',\{q\}}),
\]
so $\pi(V_r(\calA_{Q',Q'}))$ is nowhere dense. Set
\[
\calA':=\calA_{Q\setminus Q',\{q_0\}}.
\]
Since $A'$ has $n-1$ strongly connected components,  it is left to show that
\[
V_r(\calA) = V_r(\calA').
\]
Since the set on the right is included in the set on the left and both sets are closed, it suffices to show that $V_r(\calA')$ is dense in $V_r(\calA)$.

\medskip
Since $V_r(\calA)=\Gamma(f)$, it is enough to show that
\[
J \cap \pi(V_r(\calA')) \neq \emptyset
\]
for every open interval $J\subseteq [0,1]$. Suppose that this fails for some $J$. Then for every $w \in L(\calA)$ with $\pi(v_{(r,r)}(w))\in J$, the acceptance run of $w$  is eventually in $Q'$.
Pick $n\in \N$, $q \in Q'$ and $w_1\times w_2\in P_{q_0,q}(n)$ such that
\[
[v_r(w_1), v_r(w_1) + r^{-n}] \subseteq J.
\]
Set $J':= \big[v_r(w_1), v_r(w_1) + r^{-n}\big]$.

\medskip We end the proof by showing that
\[
\Gamma(f|_{J'}) \subseteq \bigcup_{u \in \Sigma_r^n} \big(v_{(r,r)}(w_1\times u) + r^{-n}V_r(\calA_{Q',Q'})\big).
\]
As $\pi(V_r(\calA_{Q',Q'}))$ is nowhere dense, this containment is a contradiction. Let $(x,y) \in \Gamma(f|_{J'})$ and $u_1 \times u_2 \in L(\calA)$ be such that $v_{(r,r)}(u_1\times u_2)=(x,y)$.  Let $s_0s_1s_2\ldots \in Q^{\omega}$ be an acceptance run of $u_1 \times u_2$. Recall that for $\sigma \in \Sigma_r^{\omega}$ we write $\sigma|_n$ for $\sigma_{1}\dots \sigma_{n}$. Note that $u_1|_n = w_1$, so we have that $v_r(z_1) \in J'$ for any word $z_1 \times z_2 \in L(\calA)$ with $z_1|_n \times z_2|_n = u_1|_n \times u_2|_n$. As $J' \subseteq J$, this means that any acceptance run of any such word $z_1\times z_2$ is eventually in $Q'$. Because $\calA$ is closed, this shows that $s_n$ is necessarily in $Q'$. Thus
\[
(x,y) \in v_{(r,r)}(w_1\times u_2|_n) + r^{-n}V_r(\calA_{Q',Q'}).\qedhere
\]
\end{proof}

\subsection{Proof of Theorem~\ref{nowhere}} Let $f: [0,1] \to [0,1]$ be a continuous $r$-regular function. By Proposition~\ref{Th57} and Lemma~\ref{cor:connectedcomp} there is a closed full deterministic B\"uchi automaton $\calA=(Q,\Sigma_r^2,\Delta,\{q_0\},F)$ that $r$-recognizes $\Gamma(f)$. Fix this automaton $\calA$ for the rest of this section.

\begin{lem}\label{lem:scaffine}
Let $Q'$ be a sink of  $\calA$. Then there is $q\in Q'$  and a $\Q$-affine function $g: [0,1] \to [0,1]$ such that $\Gamma(g) = V_r(\calA_{Q',\{q\}})$.
\end{lem}
\begin{proof}
Since $\calA$ is full, there is a state $q\in Q'$ such that
\[
\pi(V_r(\calA_{Q',\{q\}})) = [0,1].
\]
For convenience, set $\calA':= \calA_{Q',\{q\}}$.
Consider a path
\[
((s_0,\sigma_1,s_1),\dots,(s_{n-1},\sigma_n,s_n))
\]
on $\calA$ such that $s_0=q_0$ and $s_n=q$. Set $\sigma:=\sigma_1\dots \sigma_n$. Then
\[
\{ \sigma w \ : \ w \in L(\calA') \} \subseteq L(\calA).
\]
Thus
\[
v_{(r,r)}(\sigma) + r^{-n}V_r(\calA')= v_{(r,r)}\{ \sigma w \ : \ w \in L(\calA')\}   \subseteq \Gamma(f).
\]
Since $\pi(V_r(\calA'))=[0,1]$, we have that $\pi\big( v_{(r,r)}(\sigma) + r^{-n}V_r(\calA')\big)$ is a closed interval $J \subseteq [0,1]$. Therefore
\[
v_{(r,r)}(\sigma) + r^{-n}V_r(\calA') = \Gamma(f|_J)
\]
and so $V_r(\calA')$ is the graph of a continuous function. Since $\calA'$ is strongly connected, this function has to be affine by Theorem~\ref{prop:affine}.
\end{proof}


\begin{prop}\label{prop:localaffine}
Let $Q'\subseteq Q$ be a sink of $\calA$ and $q\in Q'$ be such that $\pi(V_r(\calA_{Q',\{q\}}))=[0,1]$.
Then there is a $\Q$-affine function $g: [0,1] \to [0,1]$ such that
\[
f(v_r(w_1)+r^{-n}x) = v_r(w_2) + r^{-n}g(x)
\]
 for all $n\in \N$, $w_1 \times w_2\in P_{q_0,q}(n)$, and $0 \leq x \leq 1$.

\end{prop}
\begin{proof}
By Lemma~\ref{lem:scaffine} there is a $\Q$-affine function $g: [0,1] \to [0,1]$ such that
\[
\Gamma(g) = V_r(\calA_{Q',\{q\}}).
\]
Let $w_1 \times w_2\in P_{q_0,q}(n)$. Let $x \in [0,1]$ and take $u_1 \times u_2 \in L(\calA_{Q',\{q\}})$ be such that $v_r(u_1)=x$.
 Then $(w_1\times w_2)(u_1\times u_2)\in L(\calA)$. Observe that $v_r(u_2)=g(x)$. Thus
\[
f(v_r(w_1)+r^{-n}x) = f(v_{r}(w_1u_1)) = v_r(w_2u_2) =v_r(w_2) + r^{-n} v_r(u_2) = v_r(w_2) + r^{-n}g(x).\qedhere
\]
\end{proof}

\begin{proof}[Proof of Theorem~\ref{nowhere}]
Let $Z$ be the set of all pairs $(Q',q)$ such that $Q'$ is a sink of $\calA$, $q\in Q'$ and $\pi(V_r(\calA_{Q',\{q\}}))=[0,1]$. For each pair $(Q',q)\in Z$, we obtain a $\Q$-affine function $g: [0,1] \to [0,1]$ that satisfies the conclusion of Proposition~\ref{prop:affine}.
Since $Z$ is finite, there is a finite set $S\subseteq \Q$ such that each of these $\Q$-affine functions has a slope in $S$. For each $\alpha\in S$, let $Y_\alpha$ be the set of all $x\in [0,1]$ such that there is $\varepsilon >0$ for which the restriction of $f$ to $[x - \varepsilon,x + \varepsilon]$ is an affine function with slope $\alpha$. Let $U_\alpha$ be the interior of $Y_\alpha$. Using Fact~\ref{BRW} it is easy to check that both $Y_\alpha$ and $U_\alpha$ are $r$-recognizable. Set $Y:=\bigcup_{\alpha\in S} U_\alpha$. By Proposition~\ref{prop:porous} it is left to show that 
$Y$ is dense in $[0,1]$.



\medskip
Let $J\subseteq [0,1]$ be an open interval. It is enough to find an open subinterval $J' \subseteq J$ such that the restriction of $f$ is affine on $J'$ with slope in $S$. Let $x\in J$ and $n \in \N$ be such that there are $w_1,\dots,w_n\in \Sigma_r$ satisfying $x=v_r(w_1\dots w_n0\dots)$ and $[x-r^{-(n-1)},x+r^{-(n-1)}]\subseteq J$. Observe that $x$ has two $r$-ary representations:
\[
x = \sum_{i=1}^{n} w_i r^{-i} = \sum_{i=1}^{n-1} w_i r^{-i} + (w_{n}-1) r^{-n} + \sum_{i=n+1} (r-1) r^{-i}.
\]
Let $w_{\ell}:=w_1\dots w_n0^{\omega}$ and $w_u:= w_1\dots (w_n-1){(r-1)}^\omega$. Thus there is $u\in \Sigma_r^\omega$ such that $w_\ell \times u \in L(\calA)$ or $w_u \times u \in L(\calA)$.
We treat the case when $w_\ell \times u \in L(\calA)$.
The case when $w_u \times u \in L(\calA)$ can be handled similarly.

\medskip
Let $s_0s_1\dots$ be the acceptance run of $w_\ell \times u$ on $\calA$. Then there is $w'=(w'_1\times w'_2) \in {(\Sigma_r^2)}^m$ such that there is a run of $w'$ from $s_{n-1}$ to a state $q$ in a sink $Q'$ with
$\pi(V_r(\calA_{Q',\{q\}}))=[0,1]$. Thus
\[
(w_1\times u_1)\dots (w_{n-1}\times u_{n-1})w' \in P_{q_0,q}(n-1+m).
\]
Set $z_1:=w_1\ldots w_{n-1}w_1'$ and set $z_2 := u_1\ldots u_{n-1}w_2'$. Note that
\[
|x-v_r(z_1)| \leq r^{-(n-1)}.
\]
Thus $v_r(z_1) \in J$.
By our choice of $S$, there is an affine function $g: [0,1] \to [0,1]$ with slope in $S$ such that
\[
f\big(v_r(z_1)+r^{-(n-1+m)}x\big) = v_r(z_2) + r^{-(n-1+m)}g(x)
\]
for all $0 \leq x \leq 1$.
Thus there is a nonempty open subinterval $J'$ of $J$ such that the restriction $f|_{J'}$ is affine with slope in $S$.
\end{proof}

\noindent As corollary we obtain the following sum form for continuous $r$-regular functions.

\begin{cor}\label{prop:integral}
Let $f : [0,1] \to [0,1]$ be continuous $r$-regular.
Then there are pairwise disjoint open $r$-recognizable subsets $U_1,\ldots,U_m$ of $[0,1]$ and rational numbers $\alpha_0,\alpha_1,\ldots,\alpha_m$ such that $[0,1] \setminus \bigcup_{i = 1}^{m} U_i$ is Lebesgue null and
\[
f(t) = \alpha_0 +  \sum_{i = 1}^{m} \alpha_i \calH^1( \{ x \in U_i : 0 \leq x \leq t \} ) \quad \text{for all} \quad 0 \leq t \leq 1.
\]
\end{cor}

\begin{proof}
By Fact~\ref{fact:csv}, we know that $f$ is Lipschitz.
Rademacher's theorem states that every Lipschitz function $g: U\subseteq \R^{\ell} \to \R^n$, where $U$ is open, is differentiable outside a set of Lebesgue measure 0. Thus from Rademacher's theorem we can conclude that $f'(t)$ exists for $\calH^1$-almost every $0 \leq t \leq 1$. By the fundamental theorem of Lebesgue integration
\[
f(t) = f(0) + \int_{0}^{t} f'(x) dx \quad \text{for all} \quad 0 \leq t \leq 1.
\]
Recall that $f(0)$ is rational by Lemma~\ref{lem:affine0}.
We declare $\alpha_0 = f(0)$.
By Theorem~\ref{nowhere} there are pairwise disjoint $r$-regular open subsets $U_1,\ldots,U_m$ and $\alpha_1,\ldots,\alpha_m\in \Q$ such that $W := \bigcup_{i = 1}^{m} U_i$ is dense in $[0,1]$ and $f'(x) = \alpha_i$ for all $x \in U_i$. Thus  $[0,1] \setminus W$ is Lebesgue null by Proposition~\ref{prop:porous}. Therefore, we have
\[
\int_{0}^{t} f'(x) dx = \int_{[0,t] \cap W} f'(x) dx =  \sum_{i = 1}^{m} \alpha_i \calH^1( \{ x \in U_i : 0 \leq x \leq t \})
\]
for all $0 \leq t \leq 1$.
\end{proof}

\section{Differentiability}
We are now ready to prove Theorem B as a corollary to Theorem~\ref{nowhere}. We start by proving the promised characterization of differentiable $r$-regular functions. In the proof we will use Darboux's classical theorem that for every differentiable function $g: I \to \R$ where $I$ is a closed interval, the derivative $g'$ has the intermediate value property.

\begin{thm}\label{diff}
A differentiable $r$-regular function $f:[0,1]\rightarrow\mathbb{R}$ is $\mathbb{Q}$-affine.
\end{thm}

\begin{proof}
Let $f:[0,1]\rightarrow [0,1]$ be a differentiable $r$-regular function, and suppose $\Gamma(f)$ is $r$-accepted by a closed full deterministic B\"uchi automaton $\calA$.
By Theorem~\ref{nowhere} there exist a finite set $S$ of rational numbers such that if the restriction of $f$ to a subinterval $J$ of $[0,1]$ is affine, then the slope of $f$ on $J$ is in $S$.  It follows from the proof of Theorem~\ref{nowhere} that for each sink $Q'$ of $\calA$ and $q\in Q'$ with $\pi(V_r(\calA_{Q',\{q\}}))=[0,1]$, there is an affine function $g: [0,1] \to [0,1]$ with slope in $S$ that satisfies the conclusion of Proposition~\ref{prop:localaffine}.

\medskip
We suppose towards a contradiction that $f$ is not affine. Then $f'$ is non-constant.
An application of Darboux's theorem shows $f'$ takes infinitely many values.
Fix $0 \leq x \leq 1$ such that $f'(x)\notin S$. Let $w_1 \times w_2 \in {(\Sigma_r^2)}^{\omega}$ be such that $v_{(r,r)}(w_1\times w_2) = (x,f(x))$. 
By our choice of $S$ and since $\calA$ is full, there is a final state $q$ in $\mathcal{A}$ such that $q$ is not in a sink of $\calA$, and the acceptance run $s_0s_1\dots$ of $w_1\times w_2$ in $\mathcal{A}$ visits $q$ infinitely many times. Fix this $q$. Let $\beta_i\in\mathbb{N}$ be the index of the element in the acceptance run of $(x,f(x))$ at which $q$ is visited for the $i^{th}$ time; i.e., $\beta_i$ is the $i$-th element of the set
\[
\{ k \ : \ s_k = q\}.
\]
Since $q$ is visited infinitely many times in the acceptance run of $(x,f(x))$, ${\{\beta_i\}}_{i\in\mathbb{N}}$ is cofinal in $\mathbb{N}$.
We write
\[
 x|_{\beta_i} = \sum_{j=1}^{\beta_i} w_{1,j} r^{-j} \quad \text{and} \quad f(x)|_{\beta_i} = \sum_{j=1}^{\beta_i} w_{2,j} r^{-j}.
 \]
Now let $Q'$ be a sink of $\calA$ that can be reached from $q$ and let $q'$ be a state in $Q'$ such that  $\pi(V_r(\calA_{Q',\{q'\}}))=[0,1]$. Let $g: [0,1] \to [0,1]$ be the affine function associated with $Q'$ and $q'$, and let $\alpha\in S$ be such that $g(y) = g(0) + \alpha y$ for all $y\in [0,1]$. Take $u,v \in {(\Sigma_r^2)}^k$ with $u \times v \in P_{q,q'}(k)$. Set
\[
c := \sum_{i=1}^k u_i r^{-i} \quad \text{and} \quad d:=\sum_{i=1}^k v_i r^{-i}.
\]
To compute $f'(x)$ we declare
\[
H(\beta_i,y):= \frac{f(x)-\big(f(x)|_{\beta_i} + r^{-\beta_i} d +r^{-\beta_i-k}(g(0) + \alpha y)\big)} {x-(x|_{\beta_i} + r^{-\beta_i} c +r^{-\beta_i-k} y)},
\]
where $y\in \mathbb{R}$. By our choice of $u,v$ and $\alpha$, we have that for every $y\in [0,1]$
\[
f(x|_{\beta_i} + r^{-\beta_i} c +r^{-\beta_i-k} y) = f(x)|_{\beta_i} + r^{-\beta_i} d +r^{-\beta_i-k}(g(0) + \alpha y).
\]
To further simplify notation we define
\begin{align*}
 h_1(\beta_i) :&= f(x)- \big(f(x)|_{\beta_i} + r^{-\beta_i} d +r^{-\beta_i-k} g(0)\big),\\
 h_2(\beta_i) :&= x- (x|_{\beta_i} + r^{-\beta_i} c).
\end{align*}
Thus
\[
H(\beta_i,y) =\frac{h_1(\beta_i)-r^{-\beta_i-k}\alpha y}{h_2(\beta_i)-r^{-\beta_i-k}y}=\frac{r^{\beta_i} h_1(\beta_i)-r^{-k}\alpha y}{r^{\beta_i} h_2(\beta_i)-r^{-k}y}.
\]
Note here we have
\[
-r^{-\beta_i+1}\leq  h_1(\beta_i) \leq r^{-\beta_i+1}\quad \text{and} \quad
-r^{-\beta_i+1}\leq  h_2(\beta_i) \leq r^{-\beta_i+1}.
\]
Hence $|r^{\beta_i} h_1(\beta_i)|$ and  $|r^{\beta_i} h_2(\beta_i)|$ are both bounded above by $r$.
Applying compactness we pick a subsequence ${\{ \beta_{i_{\ell}} \}}_{\ell \geq 0}$ of ${\{\beta_i \}}_{i \geq 0}$ such that the sequences $\big(r^{\beta_{i_\ell}} h_1(\beta_{i_{\ell}})\big)$ and $\big( r^{\beta_{i_\ell}} h_2(\beta_{i_{\ell}})\big)$ have limits in $\R$, which we denote by $h_1$ and $h_2$, respectively.
As $f$ is differentiable, we have
\[
f'(x) = \lim_{\ell\rightarrow \infty}H(\beta_{i_\ell},y) =  \frac{h_1-r^{-k}\alpha y}{h_2-r^{-k}y}.
\]
Thus the equation
\[
f'(x)h_2 - h_1 = r^{-k}y(f'(x) - \alpha)
\]
holds for infinitely many $y \in \mathbb{R}$. This can only happen if $\alpha = f'(x)$, which contradicts our choice of $x$.
\end{proof}

\noindent In order to prove that checking the differentiability of regular functions is in $\operatorname{PSPACE}$, we now recall an elementary fact from real analysis.

\begin{lem}\label{lem:average}
Let $f : [0,1] \to \mathbb{R}$ be continuous.
Then $f$ is affine if and only if
\[
 f \left( \frac{ x + y }{2} \right) = \frac{ f(x) + f(y) }{2} \quad \text{for all } 0 \leq x,y \leq 1.
\]
\end{lem}

\begin{cor} Given an B\"uchi automaton $\calA$ that recognizes an $r$-regular function $f:[0,1]\to [0,1]$, the problem of checking whether $f$ is differentiable, is in $\operatorname{PSPACE}$.
\end{cor}
\begin{proof}
Let $f_1 : {[0,1]}^2 \to [0,1]$ be the function that maps $(x,y)$ to $f \left( \frac{ x + y }{2} \right)$, and let $f_2 : {[0,1]}^2 \to [0,1]$ be the function that maps $(x,y)$ to $\frac{ f(x) + f(y) }{2}$. It is easy to see that $f_1$ and $f_2$ are both $r$-regular. By applying~\cite[Theorem 4 \& 5]{CSV} we can construct in polynomial time automata $\calA_1$ and $\calA_2$ that $r$-recognize $f_1$ and $f_2$. Now let $\calA_0$ be a B\"uchi automaton that recognizes
\[
\{ w_1 \times w_2 \in {(\Sigma_r^2)}^{\omega} \ : \ \exists w_3 \in \Sigma_r^{\omega} \ w_1 \times w_2 \times w_3 \in L(\calA_1) \cap L(\calA_2)\}.
\]
Again, it is not hard to see that such an automaton can be constructed in polynomial time from $\calA_1$ and $\calA_2$.
Recall that the universality problem for B\"uchi automata is $\operatorname{PSPACE}$-complete by Sistla, Vardi and Wolper~\cite{SVW}. Checking differentiability of $f$ is equivalent to checking whether $\calA_0$ accepts ${(\Sigma_r^2)}^{\omega}$, and hence in $\operatorname{PSPACE}$.
\end{proof}

\section{An application to model theory}\label{section:application}
 A study of continuous functions definable in expansions of the ordered additive group $(\mathbb{R},<,+,0)$ of real numbers was initiated in~\cite{HW-continuous}. By Fact~\ref{BRW} a continuous $r$-regular function is just a function definable (without parameters) in the expansion $\mathscr{T}_r$. Thus this enterprise is a vast generalization of the study of continuous regular functions. In this section, we will briefly recall earlier results from this inquiry and explain how the results in this paper give partial answers to some wide-open questions in this area.

\medskip
We first fix some notation. Throughout, $\mathscr{R}$ is an expansion of $(\mathbb{R},<,+,0)$, and ``definable'' without modification means ``$\mathscr{R}$-definable, possibly with parameters from $\mathbb{R}$''. A \textbf{dense $\omega$-order} is a definable subset of $\mathbb{R}$ that is dense in some nonempty open interval and admits a definable ordering with order type $\omega$.
We say that $\mathscr{R}$ is \textbf{type A} if it does not admit a dense $\omega$-order, \textbf{type C} if it defines all bounded Borel subsets of all $\mathbb{R}^k$, and \textbf{type B} if it is neither type A nor type C. It is easy to see that an expansion cannot be both type A and type C, and thus each expansion $\mathscr{R}$ of $(\mathbb{R},<,+,0)$ has to be either type A, type B or type C.

\medskip
Type A is a generalization of o-minimality and singles out structures whose definable sets are geometric-topologically tame (see~\cite{Lou} for an account of o-minimal expansions). Expansions of type C have undecidable theories and satisfy no logical tameness whatsoever. Type B expansions live in between these two extremes. Indeed, it is an easy exercise to check that $\mathscr{T}_r$ is type B. The set of all $x$ in $[0,1]$ with a finite base-$r$ expansion is definable in $\mathscr{T}_r$ and admits a definable order with order type $\omega$ (see~\cite[Section 3]{BH-Cantor}).
We regard $\mathscr{T}_r$ as the prototype of a type B expansion.  As seen in~\cite{HW-continuous} many interesting properties of $\mathscr{T}_r$ generalize to arbitrary type B expansions. Therefore the study of definable sets and definable function in type B expansions is a suitable generalization of the study of $r$-recognizable sets and $r$-
regular functions.

\medskip
Question~\ref{ques:reals1} is one of the main motivating questions in this enterprise. We say $\mathscr{R}$ is of \textbf{field-type} if there is an interval $I$, definable functions $\oplus,\otimes : I^2 \to I$, and $0_I, 1_I \in I$ such that $(I,<,\oplus,\otimes,0_I,1_I)$ is isomorphic to the ordered field $(\mathbb{R},<,+,\cdot,0,1)$ of real numbers.

\begin{qu}\label{ques:reals1}
Let $f: [0,1] \to \mathbb{R}$ be continuous and definable. If $f$ is nowhere locally affine, must $\mathscr{R}$ be of field-type?
\end{qu}

\noindent This question essentially asks whether one can recover an ordered field from a nowhere locally affine continuous function over the real additive group.
Observe that any type C expansion is trivially of field-type. Furthermore, by~\cite[Theorem B]{HW-continuous} a type A expansions that defines a nowhere locally affine function is also of field-type. Since type B expansions are not of field-type by~\cite[Fact 1.2]{HW-continuous}, Question~\ref{ques:reals1} is equivalent to the following question about type B structures.

\begin{qu}\label{ques:reals2}
Let $f: [0,1] \to \mathbb{R}$ be continuous and definable. If $\mathscr{R}$ is type B, must $f$ be locally affine on a dense open subset of $[0,1]$?
\end{qu}

%



\noindent Question~\ref{ques:reals1} (and hence Question~\ref{ques:reals2}) is one of the most important questions on definability in first-order expansions of $(\R,<,+,0)$. We give a positive answer to an interesting special case of Question~\ref{ques:reals2}.

\begin{thm}\label{thm:reals4}
Every $\mathscr{T}_r$-definable continuous function  $f : [0,1] \to \mathbb{R}$ is locally affine on a dense open subset of $[0,1]$.
\end{thm}

\noindent Observe that Theorem~\ref{thm:reals4} follows easily from Theorem A and Fact~\ref{BRW} when $f$ is definable without parameters. We just need to explain how to handle the case when $f$ is defined using additional parameters. Conversely, by Fact~\ref{BRW} we can deduce from Theorem~\ref{thm:reals4} that any continuous function $f: [0,1] \to [0,1]$ whose graph is $r$-recognized by an B\"uchi automaton with advice (i.e., by an automaton that can access a fixed infinite advice string) is locally affine away from a nowhere dense subset of $[0,1]$.

\medskip
Theorem~\ref{thm:reals4} requires the following lemma. Model-theorists might recognize it as a weak form of definable choice.

\begin{lem}\label{lem:zero-buchi}
Every nonempty subset of $\mathbb{R}^k$ that is $\mathscr{T}_r$-definable without parameters contains an element that is definable without parameters.
\end{lem}

\begin{proof}
Recall Fact~\ref{fact:periodic}: a B\"uchi automaton that accepts an infinite word accepts an infinite periodic word.
It follows from Fact~\ref{BRW} that any subset of ${[0,1]}^n$ that is definable without parameters in $\mathscr{T}_r$ has an element with rational coordinates.
Let $X \subseteq \mathbb{R}^n$ be definable without parameters in $\mathscr{T}_r$.
Fix $q \in \mathbb{Q}^n$ such that $X +q$ intersects ${[0,1]}^n$.
Then $(X + q) \cap {[0,1]}^n$ is $\mathscr{T}_r$-definable without parameters and thus contains a point $p$ with rational coordinates.
So $p - q \in X$ has rational coordinates and is thus definable without parameters.
\end{proof}

\begin{proof}[Proof of Theorem~\ref{thm:reals4}]
Let $f : [0,1] \to \R$ be continuous and $\mathscr{T}_r$-definable. After rescaling by a rational factor if necessary, we may assume that $f$ takes values in $[0,1]$.
If $f$ is definable without parameters, then $f$ is somewhere locally affine. Towards a contradiction, we now suppose $f$ is nowhere locally affine. To reach a contradiction, we will construct a parameter-free $\mathscr{T}_r$-definable function $[0,1] \to [0,1]$ that is continuous and nowhere locally affine.

\medskip
Suppose $\varphi(x,y;z)$ is a formula and $c \in \mathbb{R}^n$ is such that $\mathscr{T}_r \models \varphi(a,b;c)$ if and only if $0 \leq a,b \leq 1$ and $f(a) = b$.
Let $X$ be the set of $d \in \mathbb{R}^n$ such that
\begin{itemize}
\item $ \{ (a,b) \in {[0,1]}^2 : \mathscr{T}_r \models \varphi(a,b;d)\} $ is the graph of a continuous function $[0,1] \to [0,1]$, and
\item for all non-empty open subintervals $J$ of $[0,1]$ there are $x,y \in J$ and $0 \leq x',y' \leq 1$ such that
\[
\mathscr{T}_r \models \varphi(x,y;d), \quad \mathscr{T}_r \models \varphi(x',y';d), \quad \text{and} \quad \mathscr{T}_r \nvDash \varphi \left( \frac{ x + y }{2} , \frac{ x' + y '}{2} ;d \right).
\]
\end{itemize}
Lemma~\ref{lem:average} implies $X$ is the set of $d$ such that $ \{ (a,b) \in \mathbb{R}^2 : \mathscr{T}_r \models \varphi(a,b;d)\} $ is the graph of a continuous nowhere locally affine function $[0,1] \to [0,1]$.
Note that $X$ is $\mathscr{T}_r$-definable without parameters.
By Lemma~\ref{lem:zero-buchi} there is a parameter-free definable $c' \in X$.
Then $ \{ (a,b) \in {[0,1]}^2 : \mathscr{T}_r \models \varphi(a,b;c')\}$ is the graph of a continuous nowhere locally affine function $[0,1] \to [0,1]$ that is definable in $\mathscr{T}_r$ without parameters.
Contradiction.
\end{proof}

\noindent Of course, there are more questions that can be asked about the definable continuous function in type B expansions.
For our purposes a \textbf{space filling curve} is a continuous surjection $[0,1] \to {[0,1]}^2$.

\begin{qu}\label{ques:spacefilling}
Is there a type B expansion that defines a space filling curve?
\end{qu}
\noindent If $\mathscr{R}$ is type A, then there is no continuous definable surjection ${[0,1]}^n \to {[0,1]}^m$ when $m > n$ (see~\cite[Theorem E]{FHW-Compact}). Thus Question~\ref{ques:spacefilling} is equivalent to the following: must $(\mathbb{R},<,+,0,f)$ be type C when $f$ is a space filling curve? We give a negative answer to Question~\ref{ques:spacefilling} in a special case.

\begin{prop}
The structure $\mathscr{T}_r$ does not define a space filling curve.
\end{prop}
\begin{proof}[Sketch of proof]
An argument similar to that given in the proof above shows that if $\mathscr{T}_r$ defines a space filling curve, then there is a space filling curve that is $\mathscr{T}_r$ definable without parameters. It follows from Fact~\ref{fact:csv} that every function $f : [0,1] \to {[0,1]}^2$ that is $\mathscr{T}_r$ definable without parameters is Lipschitz. Such a function cannot be surjective, as Lipschitz maps cannot raise Hausdorff dimension.
\end{proof}

\bibliographystyle{alpha}
\bibliography{ref}

\end{document}